\documentclass{article}
\usepackage{macrodefs, amsthm, subcaption, graphicx, algorithm, algorithmic, nicematrix, authblk}
\usepackage[hidelinks]{hyperref}
\usepackage{cleveref}
\usepackage[margin=1in]{geometry}

\newtheorem{theorem}{Theorem}
\newtheorem{proposition}{Proposition}
\numberwithin{equation}{section}

\title{Data Assimilation With An Integral-Form \\ Ensemble Square-Root Filter}
\author[1]{Robin Armstrong\footnote{Email: \texttt{rja243@cornell.edu} (corresponding author)}}
\affil[1]{Cornell University, Center for Applied Mathematics, Ithaca, NY 14850, USA}
\author[2]{Ian Grooms\footnote{Email: \texttt{ian.grooms@colorado.edu}}}
\affil[2]{University of Colorado Boulder, Dept.\ of Applied Mathematics, Boulder, CO 80309, USA}
\date{}

\begin{document}
    \maketitle

    \begin{abstract}
    Geoscientific applications of ensemble Kalman filters face several computational challenges arising from the high dimensionality of the forecast covariance matrix, particularly when this matrix incorporates localization. For square-root filters, updating the perturbations of the ensemble members from their mean is an especially challenging step, one which generally requires approximations that introduce a trade-off between accuracy and computational cost. This paper describes an ensemble square-root filter which achieves a favorable trade-off between these factors by discretizing an integral representation of the Kalman filter update equations, and in doing so, avoids a direct evaluation of the matrix square-root in the perturbation update stage. This algorithm, which we call InFo-ESRF, is parallelizable and uses a preconditioned Krylov method to update perturbations to a high degree of accuracy. Through numerical experiments with both a Gaussian forecast model and a multi-layer Lorenz-type system, we demonstrate that InFo-ESRF is competitive or superior to several existing localized square-root filters in terms of accuracy and cost.
    \end{abstract}

    \section{Introduction} \label{section:intro}

\indent

In meteorology, estimating the current state of the atmosphere is a challenging task which requires synthesizing information from both empirical data and numerical models. Data assimilation \cite{daley, law_da_intro}, or DA, accomplishes this in a Bayesian manner by using empirical observations to condition a model-based prior. Ensemble Kalman filters \cite{burgers_perturbed_obs, evensen_enkf}, or 
EnKFs, are a class of DA algorithms which represent the prior in terms of a Monte-Carlo ``forecast ensemble'' of model predictions, which is transformed into an ``analysis ensemble'' representing a sample from the posterior. This paper is concerned with deterministic EnKFs, also called ensemble square-root filters \cite{whitaker_ensrf}. In this class of filters, the DA transformation takes the form of a ``mean update'' that recenters the ensemble according to the data, followed by a ``perturbation update'' that adjusts deviations from the ensemble mean to reflect the change in uncertainty.

A typical weather model will simulate $n \approx 10^9$ variables, leading to a prior covariance matrix (or ``forecast covariance'') whose array representation has $n^2 \approx 10^{18}$ entries. Meteorological applications of EnKFs must therefore contend with the fact that the forecast covariance matrix is far too large to be explicitly represented as an array. This challenge is particularly acute when the forecast uses a localized ensemble covariance, i.e., one obtained by regularizing the empirical covariance of the forecast ensemble with an \emph{a priori} assumption on spatial correlation decay \cite{hamill_distance_dependent_filtering}. This is because, unlike the raw empirical covariance, a localized ensemble covariance generally cannot be represented in low-rank form.

It is, however, possible to efficiently represent a localized ensemble covariance \emph{as a linear operator} acting on an arbitrary vector \cite{farchi_bocquet_localization}. Many other types of covariance models, including diffusion-based models \cite{weaver_courtier_diffusion_modeling} and spectrally localized ensemble covariances \cite{buehner_spectral_localization}, share this convenient property. It is therefore advantageous to design DA algorithms which interface with the forecast covariance in the ``operator access'' model rather than the ``array access'' model. The mean-update stage of a square-root filter, which entails solving a linear system involving the forecast covariance, can be readily accomplished in the operator access model using Krylov subspace methods such as conjugate gradients, MINRES, or GMRES \cite{golub_van_loan, saad_iterative_methods, trefethen_bau}. The perturbation-update stage, however, involves applying an ``adjustment'' matrix or ``modified Kalman gain'' matrix that incorporates the forecast covariance through a square-root \cite{anderson_eakf, bishop_getkf, farchi_bocquet_localization, steward_matrix_functions}. Manipulating the square-root of a matrix in the operator access model is much more challenging.

This paper develops an ensesmble square-root filter which uses a discretized integral representation of the Kalman update equations. The resulting filter updates perturbations by solving a parallelizable collection of symmetric linear systems via Krylov methods, thereby avoiding a direct evaluation of the square-root. We call this algorithm InFo-ESRF (``Integral-Form Ensemble Square-Root Filter''), and we will explain its details in \cref{section:algorithm}. Through numerical experiments in sections \ref{section:single_cycle_da} and \ref{section:cycled_da}, we show that InFo-ESRF has superior accuracy to square-root filters based on ensemble modulation \cite{bishop_getkf, farchi_bocquet_localization}. These experiments will also demonstrate that InFo-ESRF has advantages in computational efficiency when compared to square-root filters that use serial updates \cite{shlyaeva_serial_filter}. We will also show how InFo-ESRF can be enhanced with preconditioners, a property which sets it apart from existing square-root filters based on Krylov methods \cite{steward_matrix_functions}. Before this, we will review relevant background material in \cref{section:background}.

\subsection{Notation}

\indent

Vectors will be denoted by italic letters with arrows (e.g., $\vx,\, \vy$), while upper-case vectors (e.g., $\rvX,\, \rvY$) will denote vector-valued random variables. Matrices will be denoted with upper-case bold letters (e.g., $\mX,\, \mY$). \textsc{Matlab} notation will be used to write row and/or column slices of matrices. For a square matrix $\mM$ with real nonnegative eigenvalues, $\mM^{1/2}$ will denote the unique square matrix $\mS$ with real nonnegative eigenvalues such that $\mS^2 = \mM$. Note that if $\mM$ is a symmetric matrix for which $\mM^{1/2}$ exists, then $\mM^{1/2}$ is also symmetric. If in addition $\mM$ is invertible, then so is $\mM^{1/2}$, with $(\mM^{1/2})^{-1} = (\mM^{-1})^{1/2} \defeq \mM^{-1/2}$.

\subsection{Code Availability} A repository of Julia code that reproduces the numerical experiments in this paper can be downloaded from \url{https://github.com/robin-armstrong/info-esrf-experiments}.
    \section{Background} \label{section:background}

\indent

In this section we review necessary background material for this paper. \Cref{subsec:kalman} reviews basic concepts from data assimilation and Kalman filtering, and \cref{subsec:enkf} specializes to ensemble Kalman filters. \Cref{subsection:loc_intro} introduces covariance localization and the computational challenges associated with it. Existing algorithms for localized data assimilation are reviewed in \cref{subsec:localized_filters}.

\subsection{Data Assimilation} \label{subsec:kalman}

\indent

Data assimilation (DA) is a Bayesian approach for estimating the state vector $\vx \in \R^n$ of a dynamical system by combining a numerical forecast with empirical data, which may be noisy and sparse. In this framework, $\vx$ is modeled as a realization of a random-variable $\rvX_f$ on $\R^n$, whose distribution (the ``forecast distribution'') is based on the output of a numerical model and represents prior information. The observed data $\vy \in \R^d$ is modeled as a realization of a random variable $\rvY$, related to $\rvX_f$ by
\begin{equation*}
\rvY = h(\rvX_f) + \rvE_0,\quad \rvE_0 \sim \cN(\vzero,\, \mR),
\end{equation*}
where $h : \R^n \to \R^d$ is an observation operator (or ``forward operator''), $\rvE_0$ is Gaussian noise independent of $\rvX_f$ that models observation errors resulting from, e.g., instrument error, and $\mR \in \R^{d \times d}$ is an observation-error covariance matrix which we assume is nonsingular.

Given an observation consisting of a realization $\vy$ of $\rvY$, the goal of DA is to compute statistics (especially the mean and covariance) of the conditional random variable $\rvX_f \,|\, \rvY = \vy$. Geoscientific applications of DA feature extremely high system dimensions, complex forecast distributions, and nonlinear measurement operators. Under these challenging conditions, deriving practical algorithms requires simplifying approximations to be made. Two approaches are dominant: first are variational and ensemble-variational DA algorithms \cite{bannister_envar_review}, which approximate $\rvX_f$ as a Gaussian, thereby allowing the mode of $\rvX_f \,|\, \rvY = \vy$ to be estimated via a nonlinear least-squares problem. Second are Kalman filters, which consider only the first two moments of the state and observable:
\begin{equation}
\vmu \defeq \E \begin{bmatrix} \rvX_f \\ h(\rvX_f) \end{bmatrix} = \begin{bmatrix}
\vmu_x \\
\vmu_h
\end{bmatrix}\,,\quad \mSigma \defeq \cov \begin{bmatrix} \rvX_f \\ h(\rvX_f) 
\end{bmatrix} = \begin{bmatrix}
\mSigma_{xx} & \mSigma_{xh} \\
\mSigma_{xh}\tp & \mSigma_{hh} \\
\end{bmatrix}\,.\label{eqn:kalman_assumption}
\end{equation}
Kalman filters compute statistics of $\rvG \,|\, \rvY = \vy$, where $\rvG \sim \cN(\vmu,\, \mSigma)$ is a random variable on $\R^{n + d}$ representing a Gaussian approximation to the forecast distribution in joint state-observation space. It is a standard result that the marginal distribution of $\rvG \,|\, \rvY = \vy$ in state space is $\cN(\vmu_a,\, \mSigma_a)$ with
\begin{equation}
\vmu_a = \vmu_x + \mK(\vy - \vmu_h),\quad \mSigma_a = \mSigma_{xx} - \mK\mSigma_{xh}\tp, \label{eqn:kalman_filter}
\end{equation}
where
\begin{equation*}
\mK = \mSigma_{xh}(\mR + \mSigma_{hh})^{-1}
\end{equation*}
is the \emph{Kalman gain matrix}. We call $\vmu_a,\, \mSigma_a$ the \emph{Kalman filter analysis}. Note that if $\rvX_f$ and $h(\rvX_f)$ are jointly Gaussian, (equivalently, if $\rvX_f$ is Gaussian and $h$ is affine), then $\rvX_f \,|\, \rvY = \vy \sim \cN(\vmu_a,\, \mSigma_a)$ exactly.

\subsection{Ensemble Kalman Filters} \label{subsec:enkf}

\indent

An ensemble Kalman filter (EnKF) assimilates data in a Monte-Carlo fashion by generating samples from a distribution approximating that of $\rvX_f \,|\, \rvY = \vy$. Specifically, an EnKF seeks to transform samples from $\rvX_f$ into samples from a new random variable $\rvX_a$ whose first two moments match the Kalman filter analysis given by \cref{eqn:kalman_filter}. To describe the algorithmic details of an EnKF, we write
\begin{equation*}
\rvX_f = \vmu_x + \rvZ_f,
\end{equation*}
where $\vmu_x = \E[\rvX_f],\, \E[\rvZ_f] = \vzero$ and $\cov[\rvZ_f] = \cov[\rvX_f] = \mSigma_{xx}$. Conceptually, an EnKF can be separated into a ``mean update'' which transforms $\vmu_x \mapsto \vmu_a$ according to \cref{eqn:kalman_filter}, followed by a ``perturbation update'' which transforms $\rvZ_f$ into a new random variable $\rvZ_a$ such that $\E[\rvZ_a] = \vzero$ and $\cov[\rvZ_a] = \mSigma_a$. We call $\rvZ_f$ and $\rvZ_a$ the ``forecast perturbation'' and ``analysis perturbation,'' respectively.

Computationally, the mean update is relatively straightforward, involving only a linear system solve (e.g., by conjugate gradients) to evaluate the action of $\mK$ on $\vy - \vmu_h$. The perturbation update is less straightforward; three standard techniques for this step are described in \cref{theorem:ensemble_kalman}. In this theorem and henceforth, we write
\begin{equation*}
h(\rvX_f) = \vmu_h + \rvW_f,
\end{equation*}
so that $\rvW_f$ satisfies $\E[\rvW_f] = \vzero$ and $\cov[\rvW_f] = \mSigma_{hh}$.
\begin{theorem}\label{theorem:ensemble_kalman}
The following three transformations produce a random variable $\rvZ_a^{(i)}$, $i = 1,\, 2,\, 3$, satisfying $\E[\rvZ_a^{(i)}] = \vzero$ and $\cov[\rvZ_a^{(i)}] = \mSigma_a$:
\begin{equation}
\rvZ_a^{(1)} = \mA\rvZ_f, \label{eqn:adjustment_enkf}
\end{equation}
where $\mA$ is any $n \times n$ ``adjustment matrix'' satisfying $\mA\mSigma_{xx}\mA\tp = \mSigma_a$,
\begin{equation}
\rvZ_a^{(2)} = \rvZ_f - \mK(\rvW_f + \rvE), \label{eqn:stochastic_enkf}
\end{equation}
where $\rvE \sim \cN(\vzero,\, \mR)$ is independent of $(\rvX_f,\, \rvY)$, and
\begin{equation}
\rvZ_a^{(3)} = \rvZ_f - \mG\rvW_f, \label{eqn:gain_form_etkf}
\end{equation}
where
\begin{equation}
\mG = \mSigma_{xh}(\mR + \mSigma_{hh} + \mR(\mI + \mR^{-1}\mSigma_{hh})^{1/2})^{-1} \label{eqn:modified_gain}
\end{equation}
is called a \emph{modified Kalman gain matrix}.
\end{theorem}

\begin{proof}
The result for $i = 1$ follows from basic properties of random variables under linear transformations. The result for $i = 2$ is well-known in data assimilation literature, and appeared in early papers on ensemble Kalman filters \cite{burgers_perturbed_obs, houtekamer_mitchell_enkf}. The result for $i = 3$ has been used in the development of several localized filters \cite{bishop_getkf, farchi_bocquet_localization, steward_matrix_functions, whitaker_ensrf}, but to the best of our knowledge, the literature does not contain a complete proof from first principles of probability. This paper provides such a proof in the form of two results to be stated in \cref{section:algorithm}. Specifically, \cref{prop:integral_getkf_equivalence} shows that \cref{eqn:gain_form_etkf} is equivalent to an integral-form perturbation update introduced in \cref{section:algorithm}, and \cref{theorem:quadrature_trick} shows that this integral-form update yields a correctly distributed analysis ensemble.
\end{proof}

Existing filters based on \cref{eqn:adjustment_enkf}, notably the ensemble adjustment Kalman filter \cite{anderson_eakf}, use
\begin{equation}
\mA_1 = \mSigma_a^{1/2}\mSigma_{xx}^{-1/2}. \label{eqn:adjustment_matrix}
\end{equation}
A technical lemma in this paper, which we note here for its independent significance, states that
\begin{equation*}
\mA_2 = (\mSigma_a\mSigma_{xx}^{-1})^{1/2}
\end{equation*}
is also a valid adjustment matrix\footnote{Note that in general, $\mA_1 \neq \mA_2$.}. We prove this in \cref{section:adjustment_proof}. We also note that \cref{eqn:modified_gain} is not the only possible formula for a modified Kalman gain. For example, when $h(\vx) = \mH\vx$ is linear, any $\mG$ such that $\mSigma_a = (\mI - \mG\mH)\mSigma_{xx}(\mI - \mG\mH)\tp$ will suffice.

In practical applications of ensemble data assimilation, $\rvX_f$ is represented by a ``forecast ensemble.'' This is a finite sample
\begin{equation*}
\vx_f^{\,(1)},\, \ldots,\, \vx_f^{\,(m)} \in \R^n,\quad m \ll n,
\end{equation*}
usually representing an ensemble of model outputs from different initial conditions, which is transformed into a sample $\vx_a^{\,(1)},\, \ldots,\, \vx_a^{\,(m)}$, called the ``analysis ensemble,'' according to \cref{theorem:ensemble_kalman}. First, a combination of physical and statistical modeling techniques \cite{bannister_covar_1, bannister_covar_2} are used to obtain estimates of the prior statistics,
\begin{equation}
\vmu_x \approx \tildemu_x,\quad \vmu_h \approx \tildemu_h,\quad \mSigma_{xx} \approx \widetilde{\mSigma}_{xx},\quad \mSigma_{xh} \approx \widetilde{\mSigma}_{xh},\quad \mSigma_{hh} \approx \widetilde{\mSigma}_{hh}, \label{eqn:estimated_forecast_stats}
\end{equation}
which define an estimated Kalman gain $\widetilde{\mK} = \widetilde{\mSigma}_{xh}(\mR + \widetilde{\mSigma}_{hh})^{-1}$. After the mean update $\tildemu_a = \tildemu_f + \widetilde{\mK}(\vy - \tildemu_h)$, we consider the normalized forecast perturbations in state and observation space,
\begin{equation}
\vz_f^{\,(i)} = \frac{\vx_f^{\,(i)} - \tildemu_x}{\sqrt{m - 1}},\qquad \vw_f^{\,(i)} = \frac{h(\vx_f^{\,(i)}) - \tildemu_h}{\sqrt{m - 1}}. \label{eqn:normalized_perturbations}
\end{equation}
The stochastic EnKF \cite{burgers_perturbed_obs, houtekamer_mitchell_enkf}, based on \cref{eqn:stochastic_enkf}, generates i.i.d.\ samples $\rvE^{\,(i)} \sim \cN(\vzero,\, \mR)$ and updates $\vz_a^{\,(i)} = \vz_f^{\,(i)} - \widetilde{\mK}(\vw_f^{\,(i)} + \rvE^{(i)})$ for each $i$. This is usually performed simultaneously with the mean update by setting
\begin{equation*}
\vx_a^{\,(i)} = \vx_f^{\,(i)} + \widetilde{\mK}\left(\vy - h\left(\vx_f^{\,(i)}\right) - \rvE^{\,(i)}\right), \qquad i = 1,\, \ldots,\, m.
\end{equation*}
The stochastic EnKF was followed by several deterministic filters that do not use random perturbations, including the ensemble adjustment Kalman filter (EAKF) \cite{anderson_eakf}, the serial ensemble square-root filter (ESRF) \cite{whitaker_ensrf}, the ensemble transform Kalman filter (ETKF) \cite{bishop_etkf}, and the singular evolutive interpolated Kalman filter (SEIK) \cite{pham_seik}. Most relevant for this paper is the gain-form ensemble transform Kalman filter (GETKF) \cite{bishop_getkf}, which is based on \cref{eqn:gain_form_etkf}, and uses
\begin{equation}
\vz_a^{\,(i)} = \vz_f^{\,(i)} - \widetilde{\mG}\vw_f^{\,(i)} \label{eqn:getkf}
\end{equation}
where $\widetilde{\mG} = \widetilde{\mSigma}_{xh}(\mR + \widetilde{\mSigma}_{hh} + \mR(\mI + \mR^{-1}\widetilde{\mSigma}_{hh})^{1/2})^{-1}$ is an estimate of the modified Kalman gain.

\subsection{Challenges Associated with Sample Errors and Localization} \label{subsection:loc_intro}

\indent

In practical applications of ensemble Kalman filtering, complications arise from inexact knowledge of the forecast statistics, particularly the forecast covariance matrices $\mSigma_{xx},\, \mSigma_{xh}$, and $\mSigma_{hh}$. The most straightforward way to estimate these matrices is to approximate them directly from the forecast ensemble:
\begin{align*}
\widetilde{\mSigma}_{xx} &= \frac{1}{m - 1} \sum_{i = 1}^m (\vx_f^{\,(i)} - \tildemu_x)(\vx_f^{\,(i)} - \tildemu_x)\tp = \mZ\mZ\tp \\
\widetilde{\mSigma}_{xh} &= \frac{1}{m - 1} \sum_{i = 1}^m (\vx_f^{\,(i)} - \tildemu_x)(h(\vx_f^{\,(i)}) - \tildemu_h)\tp = \mZ\mW\tp \\
\widetilde{\mSigma}_{hh} &= \frac{1}{m - 1} \sum_{i = 1}^m (h(\vx_f^{\,(i)}) - \tildemu_h)(h(\vx_f^{\,(i)}) - \tildemu_h)\tp = \mW\mW\tp,
\end{align*}
where $\tildemu_x = \frac{1}{m} \sum_{i = 1}^m  \vx_f^{\,(i)}$, $\tildemu_h = \frac{1}{m}\sum_{i = 1}^m h(\vx_f^{\,(i)})$, and
\begin{equation*}
\mZ = [ \vz_f^{\,(1)} \: \cdots \: \vz_f^{\,(m)} ] \in \R^{n \times m},\qquad \mW = [ \vw_f^{\,(1)} \: \cdots \: \vw_f^{\,(m)} ] \in \R^{d \times m}
\end{equation*}
are normalized forecast perturbation matrices. Because the ensemble size $m$ is generally much less than the model dimension $n$, these estimates are extremely rank deficient and contain significant sampling errors. The problem is illustrated in \cref{fig:covar_comparison}, which compares a full-rank covariance matrix to an ensemble estimate. The low-rank structure of the ensemble estimate gives rise to noise far off the diagonal, creating the appearance of correlations between model variables that are widely separated in space. Assimilating data using this estimate would result in spurious updates to model variables unrelated to those being measured, generally leading to an analysis ensemble with unrealistically low variance.

\begin{figure}
    \centering
    \includegraphics[scale=.7]{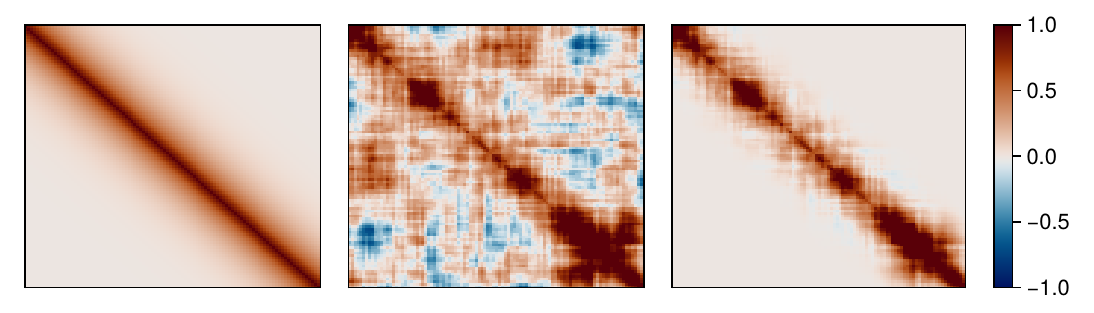}
    \caption{Left: a $100 \times 100$ covariance matrix $\mSigma_{xx}$. Middle: an ensemble estimate $\widetilde{\mSigma}_{xx} = \mZ\mZ\tp$, where $\mZ$ is derived from 20 independent samples from $\cN(\vzero,\, \mSigma_{xx})$. Right: a localized estimate $\widehat{\mSigma}_{xx} = \mL \circ (\mZ\mZ\tp)$, where $\mL(i,j) = \cG(|i - j|)$ and $\cG$ is the Gaspari-Cohn localization function \cite{gaspari_cohn} of radius 20.}
    \label{fig:covar_comparison}
\end{figure}
Localization is a class of techniques for mitigating the harmful effects of forecast covariance rank deficiency. For example, a large class of ``observation space'' methods address rank deficiency by localizing the impact of a new observation to model variables within some radius of the measurement location; see, e.g., \cite{hamill_distance_dependent_filtering, haugen_evensen_indian_ogcm}. These techniques may perform poorly when the measurement is not clearly associated with a single point in space; for example, satellite radiance measurements observe a weighted vertical integral across an entire atmospheric column \cite{bishop_getkf}. ``Model-space localization'' forgoes assumptions on observation locality, and instead estimates the state-space forecast covariance as
\begin{equation}
\widetilde{\mSigma}_{xx} = \widehat{\mSigma}_{xx} \defeq \mL \circ (\mZ\mZ\tp), \label{eqn:localized_state_covar}
\end{equation}
where $\circ$ denotes an element-wise product, and $\mL \in [0,1]^{n \times n}$ is a symmetric positive semidefinite matrix whose $(i,j)$ entry attenuates correlations between variables $\rvX_i$ and $\rvX_j$. \Cref{fig:covar_comparison} shows how this can lead to a more realistic approximation of the overall forecast covariance. Often $\mL$ is constructed from the Gaspari-Cohn localization function \cite{gaspari_cohn} which attenuates correlations according to the degree of spatial separation. The Schur product theorem \cite[Thm 7.5.3]{horn_johnson} guarantees that $\widehat{\mSigma}_{xx}$ is positive-semidefinite, and is therefore mathematically valid as a covariance approximation.

Finally, if $h$ is differentiable, then localized ensemble approximations of $\mSigma_{xh}$ and $\mSigma_{hh}$ can be obtained by linearizing $h$, yielding
\begin{equation*}
\widehat{\mSigma}_{xh} \defeq \widehat{\mSigma}_{xx}\mH\tp,\quad \widehat{\mSigma}_{hh} \defeq \mH\widehat{\mSigma}_{xx}\mH\tp,
\end{equation*}
where $\mH(i,j) = \partial_j h_i(\tildemu_x)$.

\subsection{Localized Square-Root Filters} \label{subsec:localized_filters}

\indent

While model-space localization produces full-rank forecast covariances and effectively filters out spurious correlations, the fact that $\widehat{\mSigma}_{xx},\, \widehat{\mSigma}_{xh},\, \widehat{\mSigma}_{hh}$ do not have low-rank ensemble representations introduces a new set of computational challenges. The perturbation update stage of a square-root filter becomes especially difficult in the absence of an ensemble representation, owing to the high computational cost of evaluating the square-root of a large full-rank matrix. Many techniques have been developed to update perturbations without evaluating the full square-root, and for this paper, the most relevant techniques are as follows.
\begin{itemize}
\item \underline{Serial data assimilation:} When observation errors are uncorrelated, vector-valued data can be assimilated through a sequence of serial updates that each assimilate one scalar component. Shlyaeva, Snyder, and Whitaker \cite{shlyaeva_serial_filter} implemented an ensemble square-root filter with model-space localization in this manner, with the advantage that the matrix inverses and inverse square-roots that define the Kalman gain and modified Kalman gain matrices reduce to trivial scalar computations. The presence of model space localization means that this filter is noncommutative with respect to the order in which data components are assimilated \cite{steward_matrix_functions}.

\item \underline{Ensemble augmentation:} These techniques involve constructing an ensemble approximation of the localized covariance,
\begin{equation*}
\widehat{\mSigma}_{xx} \approx \mZ_*\mZ_*\tp,
\end{equation*}
where $\mZ_* \in \R^{n \times k}$ is a matrix of ``augmented forecast perturbations,'' and where the augmented ensemble size $k$ must be small enough to allow efficient computation, but large enough to preserve the effects of localization. The GETKF can then be efficiently run with model-space localization by using the empirical covariances of the augmented ensemble in state and observation space to form the Kalman gain and modified Kalman gain matrices \cite{bishop_getkf, farchi_bocquet_localization}. Bishop et al.\ \cite{bishop_getkf} used a ``modulation" technique \cite{BH09b, buehner_background_covar} based on a precomputed spectral decomposition of $\mL$ to construct $\mZ_*$. While relatively inexpensive, this produces a far-from-optimal approximation of $\widehat{\mSigma}_{xx}$ \cite{farchi_bocquet_localization}. Farchi and Bocquet \cite{farchi_bocquet_localization} used a randomized SVD \cite{halko_finding_structure_with_randomness} to construct $\mZ_*$ such that $\mZ_*\mZ_*\tp$ approximates $\widehat{\mSigma}_{xx}$ nearly optimally for the given number of augmented ensemble members. This approach uses the fact that $\widehat{\mSigma}_{xx}$ has an efficient, exact representation as a linear operator via the formula
\begin{equation}
\widehat{\mSigma}_{xx}\vu = \sum_{i = 1}^m \vz_f^{\,(i)} \circ \ell(\vz_f^{\,(i)} \circ \vu), \label{eqn:farchi_bocquet_matvec}
\end{equation}
where $\ell(\vv) = \mL\vv$ \cite{farchi_bocquet_localization}. Note that $\mL$ frequently posesses special structure that permits an efficient matrix-free implementation of $\ell$.

\item \underline{Krylov subspace methods:} If $h$ is composed with a whitening transformation so that $\mR = \mI$, then the main difficulty in the localized GETKF consists of evaluating the operator actions of $f_1(\widehat{\mSigma}_{hh})$ and $f_2(\widehat{\mSigma}_{hh})$, where
\begin{equation*}
f_1(x) = \frac{1}{1 + x},\quad f_2(x) = \frac{1}{1 + x + \sqrt{1 + x}}
\end{equation*}
act on the eigenvalues of $\widehat{\mSigma}_{hh}$. Steward et al.\ \cite{steward_matrix_functions} demonstrated a GETKF which performs these evaluations with Krylov subspace methods \cite{golub_van_loan, saad_iterative_methods, trefethen_bau}. These methods approximate matrix-vector products of the form $f(\mM)\vz$, where $\mM \in \R^{d \times d}$, $\vz \in \R^d$, and $f$ acts on the eigenvalues of $\mM$, using only operator access to $\mM$. Because $\widehat{\mSigma}_{xx}$ (and therefore $\widehat{\mSigma}_{hh}$) have exact operator representations through \cref{eqn:farchi_bocquet_matvec}, Krylov methods allow for a more accurate perturbation update than would result from approximating $\widehat{\mSigma}_{hh}$ with, e.g., ensemble augmentation.

Krylov methods compute $p(\mM)\vz$, where $p$ is a polynomial approximating $f$ in a suitable manner. This is done by evaluating $f$ within the lower-dimensional \emph{Krylov subspace}
\begin{equation*}
\cK_k(\mM,\, \vz) \defeq \vspan\{ \vz,\, \mM\vz,\, \ldots,\, \mM^{k-1}\vz \},
\end{equation*}
where $k = 1 + \deg p \ll d$. This subspace is constructed iteratively, where each iteration requires a single matrix-vector product with $\mM$. The particular Krylov method used by Steward et al.\ \cite{steward_matrix_functions} makes the approximation
\begin{equation*}
f(\mM)\vz \approx \mV_k f(\mT_k) \mV_k\tp\vz,
\end{equation*}
where the columns of $\mV_k$ are an orthonormal basis for $\cK_k(\mM,\, \vz)$ and $\mT_k = \mV_k\tp\mM\mV_k$. This is equivalent to computing $p(\mM)\vz$, where $p$ is the degree-$(k-1)$ polynomial interpolation of $f$ on the spectrum of $\mT_k$ \cite[Thm 13.5]{higham_functions_of_matrices}. For this reason, using this method successfully requires that the eigenvalues of $\mT_k$ accurately approximate those of $\mM$.
\end{itemize}

    \section{An Integral-Form Ensemble Square-Root Filter} \label{section:algorithm}

\indent

This section presents the main algorithmic contribution of this paper, a filter which we call InFo-ESRF (``Integral Form Ensemble Square-Root Filter''). The meaning of ``InFo'' will be made clear below, while ``ESRF'' honors the ESRF of \cite{whitaker_ensrf}, which was the first filter to use a modified Kalman gain.

InFo-ESRF is based on the observation that the mean-update stage of a square-root filter uses relatively straightforward computational machinery, namely a symmetric linear system solve, to apply the (standard) Kalman gain matrix. In contrast, deterministic perturbation update methods require more challenging computations to apply the matrix square-roots in the modified Kalman gain and ensemble adjustment matrices; see \cref{eqn:modified_gain} and \cref{eqn:adjustment_matrix}. InFo-ESRF avoids these complications by reformulating the perturbation update in terms of symmetric linear system solves, as explained below in \cref{subsec:info_esrf_theory,subsec:info_esrf_algorithm}. In doing so, it mitigates several challenges present in existing localized filters, including the following.
\begin{enumerate}
\item Augmentation methods (cf.\ \cref{subsec:localized_filters}) approximate the action of $\widehat{\mSigma}_{xx}$ on a $k$-dimensional subspace defined by the augmented ensemble, but even optimally approximating the action of $\widehat{\mSigma}_{xx}$ for a given $k$ does not guarantee that the action of the modified Kalman gain will be accurately recovered. Formulating the perturbation update in terms of symmetric linear system solves enables InFo-ESRF to use the machinery of Krylov subspace methods, which can directly approximate the action of a Kalman gain matrix with stronger accuracy guarantees.

\item Krylov methods for localized ensemble data assimilation \cite{steward_matrix_functions} have, thus far, involved approximating matrix functions of $\mR^{-1/2}\widehat{\mSigma}_{hh}\mR^{-1/2}$ within a Krylov subspace, and thus require the subspace to capture a high-fidelity approximation of this matrix's eigenvalues (cf.\ \cref{subsec:localized_filters}). These methods cannot use preconditioning to accelerate the Krylov iteration, as preconditioners intentionally distort a matrix's eigenvalue spectrum. In contrast, InFo-ESRF, which uses Krylov subspaces only to perform symmetric linear system solves, can be used with a preconditioner.
\end{enumerate}

\subsection{Theoretical Foundations} \label{subsec:info_esrf_theory}

\indent

The perturbation update procedure used by InFo-ESRF can be understood in analogy with the stochastic EnKF, which updates perturbations with a standard Kalman gain matrix while injecting additive noise into the observation vector (cf.\ \cref{subsec:enkf}). InFo-ESRF also updates perturbations using a standard Kalman gain, but injects \emph{multiplicative} noise in the form of a scalar inflation applied to the observation error covariance $\mR$. More precisely, perturbations are updated using the gain matrix
\begin{equation}
\mK(s) = \mSigma_{xh}((s + 1)\mR + \mSigma_{hh})^{-1}, \label{eqn:inflated_kalman_gain}
\end{equation}
which is the standard Kalman gain with observation errors inflated by $s > 0$. \Cref{theorem:stochastic_inflation} states that, when averaged over an appropriate distribution of inflation values, this procedure yields an analysis ensemble with the required covariance.
\begin{theorem} \label{theorem:stochastic_inflation}
Define a probability density $p(s) = [\pi \sqrt{s}(s+1)]^{-1}$ over $s \in (0,\infty)$. If
\begin{equation}
\rvZ_a = \rvZ_f - \int_0^\infty p(s)\mK(s)\rvW_f\,ds, \label{eqn:stochastic_obserr_inflation}
\end{equation}
then $\E[\rvZ_a] = \vzero$ and $\cov[\rvZ_a] = \mSigma_a$.
\end{theorem}

\begin{proof}
See \cref{section:integral_proof}.
\end{proof}

Recall that $\mSigma_a$, referenced in the statement of \cref{theorem:stochastic_inflation} above, denotes the Kalman filter analysis covariance matrix from \cref{eqn:kalman_filter}.

InFo-ESRF updates perturbations by evaluating a quadrature approximation to the integral in \cref{eqn:stochastic_obserr_inflation}. Because the distribution $p$ is extremely heavy-tailed, approximating this integral efficiently requires it to be reweighted in a manner that depends on the structure of the forecast covariance. \Cref{theorem:quadrature_trick} describes a general reweighting technique, and to state the theorem we define
\begin{equation*}
\mC = \mR^{-1/2}\mSigma_{hh}\mR^{-1/2},
\end{equation*}
representing the prior observable covariance normalized by observation error.
\begin{theorem} \label{theorem:quadrature_trick}
Let $\cI \subseteq \R$ be an interval, and let $r_t,\, s_t \geq 0$ be continuous functions of $t \in \cI$ such that
\begin{equation*}
\int_{t \in \cI} \frac{r_t}{s_t + c + 1}\, dt = \frac{1}{\sqrt{c+1}}
\end{equation*}
for all $c$ in some open set containing $\lambda(\mC) \union \{ 0 \}$. Define $p_t = r_t(s_t + 1)^{-1}$. Then $\int_{t \in \cI} p_t\, dt = 1$, and if
\begin{equation}
\rvZ_a = \rvZ_f - \int_{t \in \cI} p_t \mK(s_t)\rvW_f\, dt, \label{eqn:integral_form_update}
\end{equation}
then $\E[\rvZ_a] = \vzero$ and $\cov[\rvZ_a] = \mSigma_a$.
\end{theorem}

\begin{proof}
See \cref{section:integral_proof}.
\end{proof}

Note that, as demonstrated in its proof, \cref{theorem:quadrature_trick} includes \cref{theorem:stochastic_inflation} as a special case corresponding to $\cI = (0,\infty)$, $r_t = (\pi\sqrt{t})^{-1}$ and $s_t = t$.

We end this section by remarking that the integral-form perturbation update described here is mathematically equivalent to the GETKF perturbation update that uses a modified Kalman gain matrix, as shown below in \cref{prop:integral_getkf_equivalence}.

\begin{proposition} \label{prop:integral_getkf_equivalence}
Under the assumptions of \cref{theorem:quadrature_trick},
\begin{equation*}
\int_{t \in \cI} p_t \mK(s_t)\, dt = \mG,
\end{equation*}
where $\mG$ is the modified Kalman gain matrix defined in \cref{eqn:modified_gain}.
\end{proposition}

\begin{proof}
See \cref{section:modified_kalman_gain}.
\end{proof}

Despite this equivalence, we will see in the remainder of this paper that the integral form of the perturbation update affords advantages in terms of accuracy and computational cost which are not shared by the traditional GETKF.

\subsection{Algorithm Details} \label{subsec:info_esrf_algorithm}

\indent

InFo-ESRF updates the ensemble mean using the standard transformation
\begin{equation}
\tildemu_a = \tildemu_f + \widehat{\mK}(\vy - \tildemu_h), \label{eqn:info_esrf_mean_solve}
\end{equation}
where $\widehat{\mK} = \widehat{\mSigma}_{xh}(\mR + \widehat{\mSigma}_{hh})^{-1}$ is the localized ensemble approximation of the Kalman gain matrix. We apply the inverse in $\widehat{\mK}$ using preconditioned conjugate gradients, or PCG \cite[Alg 11.5.1]{golub_van_loan}, with a preconditioner described in \cref{section:preconditioning}.

To update the ensemble perturbations, InFo-ESRF begins by discretizing the integral in \cref{eqn:integral_form_update} using a $Q$-point quadrature rule with nodes $s^{(1)},\, \ldots,\, s^{(Q)} \geq 0$ and weights $p^{(1)},\, \ldots,\, p^{(Q)} \geq 0$. Formulas for $(s^{(q)},\, p^{(q)})$ based on \cref{theorem:quadrature_trick} are given in \cref{section:quad_rule}. Perturbations are then updated as follows:
\begin{equation*}
\vz_a^{\,(i)} = \vz_f^{\,(i)} - \sum_{q = 1}^Q p^{(q)} \widehat{\mK}(s^{(q)})\vw_f^{\,(i)} = \vz_f^{\,(i)} - \widehat{\mSigma}_{xh} \sum_{q = 1}^Q p^{(q)} \vv_q^{\,(i)},
\end{equation*}
where $\widehat{\mK}(s) = \widehat{\mSigma}_{xh}((s+1)\mR + \widehat{\mSigma}_{hh})^{-1}$, and $\vv_q^{\,(i)}$ ($1 \leq i \leq m,\, 1 \leq q \leq Q$) are found by solving
\begin{equation}
((s^{(q)} + 1)\mR + \widehat{\mSigma}_{hh})\vv_q^{\,(i)} = \vw_f^{\,(i)}.
\label{eqn:info_esrf_perturbation_solve}
\end{equation}
These systems are again solved with PCG, using the same preconditioner as for the mean (cf.\ \cref{section:preconditioning}). Note that these linear solves can, and should, be parallelized across $i$ and $q$. Two other potential methods of solving \eqref{eqn:info_esrf_perturbation_solve} deserve mention, although we do not implement them in this paper.
\begin{enumerate}
\item If $h$ is composed with a whitening transform so that $\mR = \mI$, then for each $i$, one may compute $\vv_1^{\,(i)},\, \ldots,\, \vv_Q^{\,(i)}$ in parallel with a \emph{single} CG process using the shift-invariance property of Krylov subspaces \cite{datta_saad_shifted_solves, frommer_glassner_shifted_solves}. While this reduces computational complexity as a function of $Q$, it also requires the user to forgo preconditioning, or to restrict themselves to a specialized class of preconditioners that preserve Krylov subspaces under scalar shifts \cite{meerbergen_shifted_solves}.
\item Block CG iteration \cite{oleary_block_cg} can be used to compute $\vv_q^{\,(1)},\, \ldots,\, \vv_q^{\,(m)}$ in parallel for each $q$, such that fewer iterations per ensemble member are required to meet a given error tolerance, by approximating $\vv_q^{\,(i)}$ within the \emph{sum} of the Krylov subspaces generated for $\vv_q^{\,(1)},\, \ldots,\, \vv_q^{\,(m)}$. While this enriches the space of possible approximations for each ensemble member, it also requires more floating-point operations per iteration. If the ensemble perturbations are updated on parallel processors, then the block CG technique also requires synchronization and data-sharing between processes, whereas standard CG admits a parallel implementation with no communication needed.
\end{enumerate}

The remainder of this paper experimentally compares InFo-ESRF with existing localized filters.

    \section{Experiment 1: Synthetic Forecast Covariance} \label{section:single_cycle_da}

\indent

Our first experiment evaluates the accuracy of the analysis ensemble covariance produced by InFo-ESRF in comparison to existing localized square-root filters. We also evaluate filters in terms of wall-clock runtime. The following filters are compared with InFo-ESRF.
\begin{enumerate}
\item A localized serial ESRF \cite{shlyaeva_serial_filter}.
\item A GETKF with ensemble augmentation via modulation \cite{bishop_getkf}.
\item A GETKF with ensemble augmentation via a randomized SVD \cite{farchi_bocquet_localization}.
\item A GETKF which evaluates the modified Kalman gain matrix within a Krylov subspace \cite{steward_matrix_functions}.
\end{enumerate}
Each of these filters is reviewed in \ref{subsec:localized_filters}. For notes on how each filter was implemented, refer to \cref{section:filter_implementation}.

\subsection{Model and Observing System} \label{subsec:single_cycle_model}

\indent

The forecast distribution is modeled as an $n$-dimensional multivariate Gaussian $\cN(\vzero,\, \mSigma_{xx})$ with $n = 2000$. The forecast covariance is given by
\begin{equation}
\mSigma_{xx}(i,j) = \delta_{ij} \eta + \exp\left( -\frac{c(i,\,j)^2}{2\nu^2} \right),\quad c(i,j) = \frac{n}{\pi}\sin\left( \frac{\pi|i - j|}{n} \right), \label{eqn:single_cycle_covar}
\end{equation}
where $\delta$ is the Kronecker delta function, $\nu = 10$ is the correlation lengthscale, and $\eta = 10^{-4}$ is a noise floor added to make $\mSigma_{xx}$ numerically full-rank. Note that $c(x,x')$ represents chordal distance between points on a circle of circumference $n$. \Cref{fig:single_cycle_covar} plots the leading part of the eigenvalue spectrum for this covariance matrix.

We use a linear forward operator $\mH$ that observes $d = 100$ weighted sums over the state vector, and these sums are defined by the matrix entries
\begin{equation}
\mH(i,j) = \exp\left(-\frac{c(j,\,20i)^2}{2b^2}\right),\quad 1 \leq i \leq d,\quad 1 \leq j \leq n, \label{eqn:single_cycle_obsop}
\end{equation}
with channel bandwidth $b = 10$. For the observation error covariance we take
\begin{equation}
\mR = r^2\mI_d, \label{eqn:single_cycle_obserr_covar}
\end{equation}
where $r^2 \approx 36.3$ is 10\% of the underlying observable variance. \Cref{fig:single_cycle_covar} shows changes in variance from forecast to analysis under this observing system.
\begin{figure}
    \centering
    \includegraphics[scale=.8]{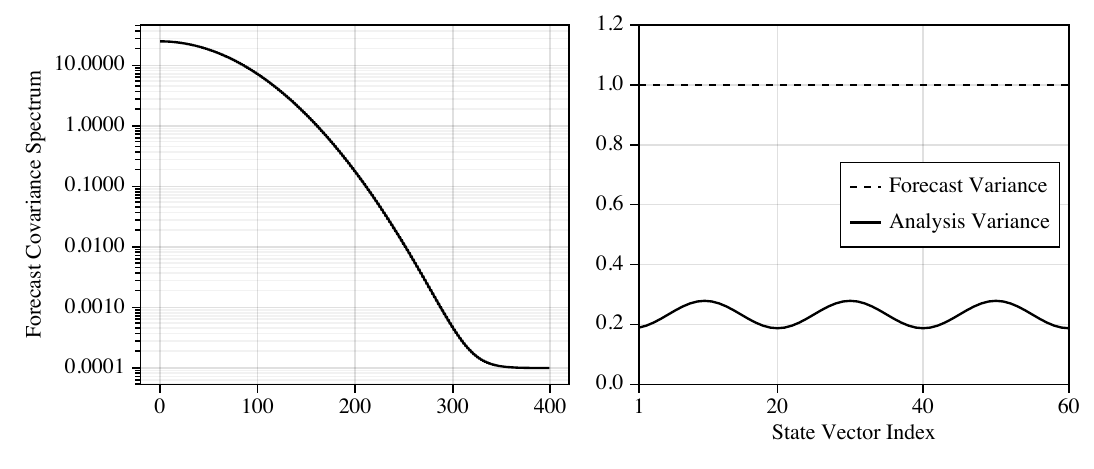}
    \caption{Left: leading entries of the eigenvalue spectrum for $\mSigma_{xx}$ defined in \cref{eqn:single_cycle_covar}. Right: forecast and analysis variances for the observing system given by \cref{eqn:single_cycle_obsop,eqn:single_cycle_obserr_covar}, over a subset of the state vector. The local minima in analysis variance correspond to peaks in the weighting functions that define each observation channel.} \label{fig:single_cycle_covar}
\end{figure}

\subsection{Localization}

\indent

To reduce spurious correlations between state variables $\rvX_i$ and $\rvX_j$, we use the localizing value
\begin{equation*}
\mL(i,j) = \exp\left( -\frac{c(i,j)^2}{2L^2} \right),
\end{equation*}
where $L = 12$. This value of $L$ minimizes the average error in the variance of an analysis ensemble produced by the GETKF with $m = 20$ ensemble members, with error quantified by \cref{eqn:diagonal_error_metric}. Average errors for tuning $L$ are measured over 20 trials. To isolate the effects of $L$ from filter-specific approximations, we tune $L$ using an idealized GETKF that forms the exact Kalman and modified Kalman gain matrices through ``brute force'' dense matrix computations.

\subsection{Experiment}

\indent

We measure filter performance over variations in parameters $p$ and $k$. The first parameter, $p$, is the number of Ritz vectors used to form the preconditioner used to update perturbations in InFo-ESRF; see \cref{section:preconditioning} for details. The second parameter, $k$, controls a trade-off between accuracy and computational cost. For the two augmentation methods \cite{bishop_getkf, farchi_bocquet_localization}, $k = \hat{m}/m$, where $\hat{m}$ is the augmented ensemble size and $m = 20$ is the original ensemble size. InFo-ESRF does not use an augmented ensemble, but does apply inflated Kalman gains to $\hat{m} \defeq Qm$ perturbation vectors, where $Q$ is the number of quadrature nodes. Therefore, for InFo-ESRF, we take $k = \hat{m}/m$ to be the quadrature size. For the serial ESRF \cite{shlyaeva_serial_filter} and the Krylov-based GETKF, we report constant averages that do not depend on $p$ or $k$.

To highlight the effects of preconditioning in InFo-ESRF, all Krylov methods (in InFo-ESRF and the Krylov-based GETKF \cite{steward_matrix_functions}) are run for only 2 iterations. We have found both filters to be remarkably effective even with such a small number of iterations. A higher number will be used for the cycled data assimilation experiment in \cref{section:cycled_da}.

Each filter assimilates data into a forecast ensemble consisting of $m = 20$ independent draws from $\cN(\vzero,\, \mSigma_{xx})$, where the ``ground-truth'' state is a $21^\mathrm{st}$ independent draw from $\cN(\vzero,\, \mSigma_{xx})$. This is repeated for 100 independent trials. For each trial we record the wall-clock runtime, as well as a measure of error in the analysis variance:
\begin{equation}
E^2 = \frac{1}{n}\sum_{i = 1}^n \frac{(\widetilde{\mSigma}_a(i,i) - \mSigma_a(i,i))^2}{\mSigma_a(i,i)^2}, \label{eqn:diagonal_error_metric}
\end{equation}
where $\widetilde{\mSigma}_a$ is the empirical covariance of the analysis ensemble, and $\mSigma_a$ is the analysis covariance derived analytically from the Kalman filter equations. For this experiment, we do not apply any inflation. Runtimes were measured using Julia 1.10.2 with OpenBLAS on a MacBook Air with an M2 chip and 8GB memory.

\subsection{Results}

\begin{figure}
    \centering
    \includegraphics[scale=.9]{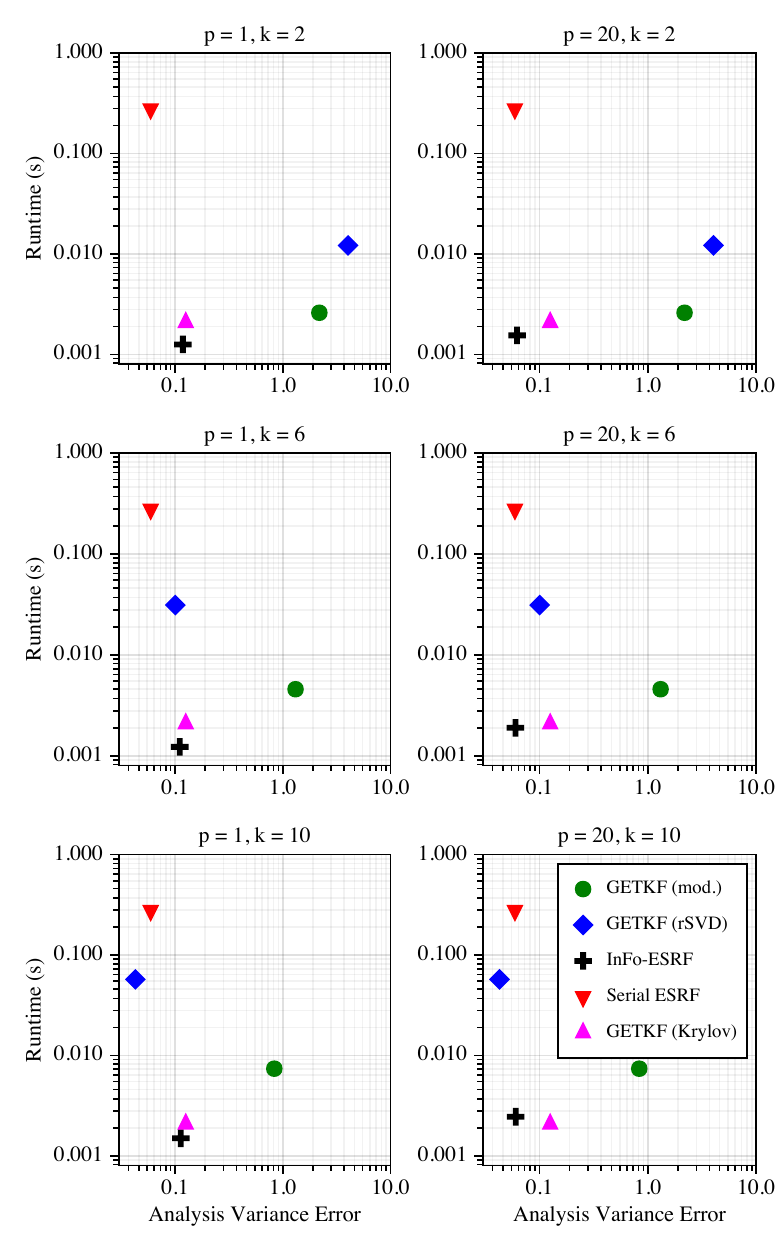}
    \caption{Results of the synthetic forecast experiment. Horizontal error bars for 95\% confidence are too narrow to visually resolve in this experiment, and are therefore not shown. Vertical error bars are similarly narrow for all filters except the serial ESRF and GETKF with randomized SVD-based modulation, for which they stretch into the negative.} \label{fig:single_cycle}
\end{figure}

\indent

The results of this experiment are shown in \cref{fig:single_cycle}. When substantial preconditioning is used ($p = 20$), InFo-ESRF exceeds all other filters in accuracy, except for the GETKF with randomized SVD-based augmentation for $k = 10$. Notably, InFo-ESRF is more accurate in this test case than the Krylov-based GETKF, which is the only other filter in this suite to use a Krylov subspace method. This illustrates the role of preconditioning in making InFo-ESRF an effective filter even for very few iterations of its Krylov method, whereas the Krylov-based GETKF uses a perturbation update that is not compatible with preconditioning (cf.\ \cref{section:algorithm}). The GETKF with modulation-based augmentation has large errors, owing to the difficulty of approximating the relatively high-rank localizing matrix $\mL$ with a small number of principal components.

The GETKF with modulation via a randomized SVD also has strikingly large errors for $k = 2$, where it approximates $\mSigma_{xx}$ at rank $km - 1 = 79$. This reflects the fact that $\mSigma_{xx}$ has a rather slowly decaying eigenvalue spectrum (cf.\ \cref{fig:single_cycle_covar}), which imposes a fundamental limit on how well it can be approximated with a small augmented ensemble, even when that ensemble is computed optimally. Interestingly, modulation produces a slightly more accurate approximation for $k = 2$, but not for higher values of $k$.

Across all test-cases, InFo-ESRF has the shortest runtime, which can be explained by several factors. It does not need to perform matrix operations over a large augmented ensemble, as the augmented-based GETKFs do. The GETKF based on a randomized SVD is particularly slow, owing mainly to the thin-QR factorization step within the random SVD \cite{halko_finding_structure_with_randomness}. InFo-ESRF also does not need to compute an eigendecomposition for each ensemble member, which is required in the Krylov-based GETKF. Finally, the serial ESRF is significantly slower than all other filters in this experiment, as its serial nature means that it must loop individually over each component of the data vector, whereas the other filters considered here manipulate the data in a single block.
    \section{Experiment 2: Cycled Data Assimilation} \label{section:cycled_da}

\indent

Our second experiment compares filter performance in terms of forecast accuracy and ensemble variance over a large number of forecast-assimilation cycles, in a manner that is conceptually similar to an operational weather prediction cycle.

\subsection{Model and Observing System}

\indent

The system being forecasted is an idealized 2D ``atmosphere'' which we have adopted from Farchi and Bocquet \cite{farchi_bocquet_localization}. Letting $X_{i,\, j}$ denote the state at column $i$ and layer $j$, the state evolves according to a coupled variant of the Lorenz '96 system \cite{lorenz_96}:
\begin{equation}
\begin{split}
\frac{dX_{i,\, j}}{dt} &= X_{i - 1,\, j}(X_{i + 1,\, j} - X_{i - 2,\, j}) - X_{i,\, j} + F_i \\
&\quad+ \delta_{j > 1} \gamma(X_{i,\, j - 1} - X_{i,\, j}) + \delta_{j < N_z}\gamma(X_{i,\, j + 1} - X_{i,\, j}),
\end{split} \label{eqn:ll96_ode}
\end{equation}
where $\gamma$ is a coupling term between adjacent horizontal layers, $F_i$ are forcing terms, and the horizontal index $i$ is periodic. We take $1 \leq i \leq N_h = 40$ and $1 \leq j \leq N_z = 32$, which gives a total model dimension of $n = 1280$. Following \cite{farchi_bocquet_localization}, we set $\gamma = 1$, and the forcing terms $F_i$ decrease linearly from $8$ at the bottom layer to $4$ at the top layer. \Cref{fig:system} plots a typical state vector for this system and illustrates the time evolution of a single vertical column.
\begin{figure}
    \centering
    \includegraphics[scale=.9]{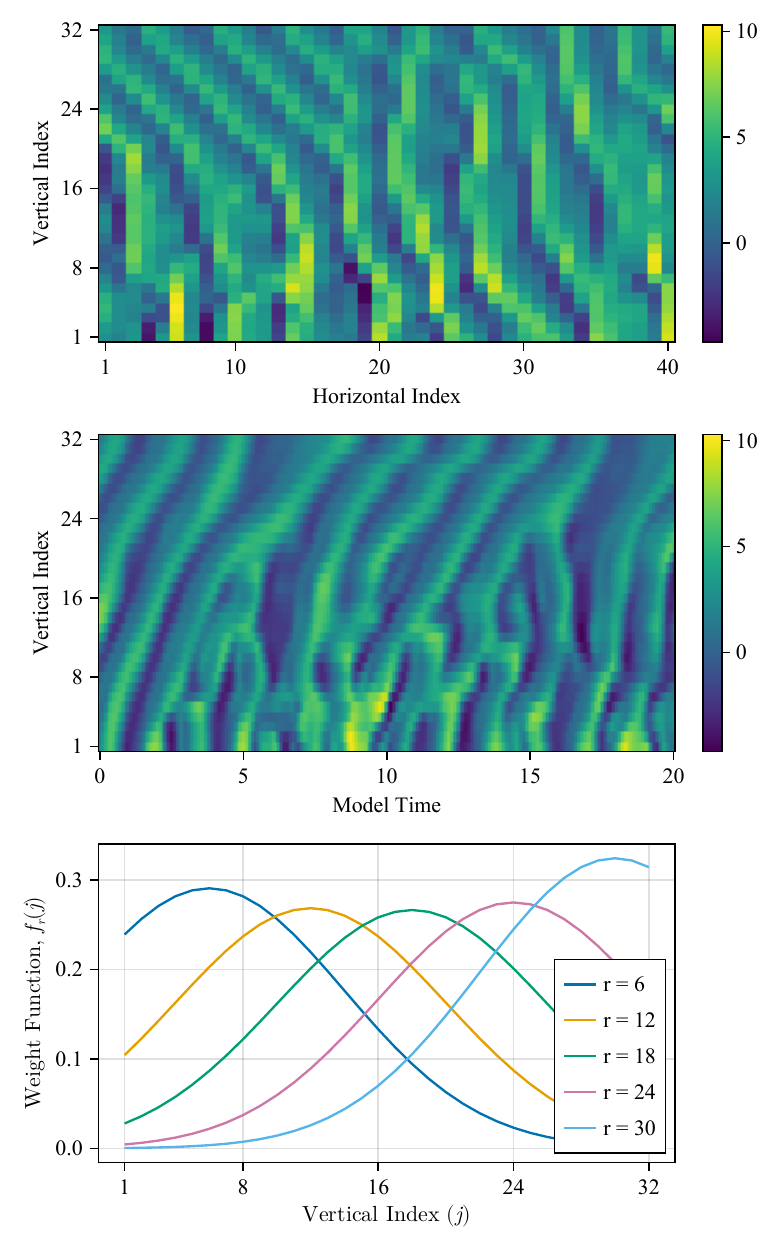}
    \caption{Top: a typical state vector for the model ``atmosphere.'' Middle: a Hovm\"oller diagram showing the time evolution of a single column. Bottom: weighting functions for the 5 channels observed for a each measured column.} \label{fig:system}
\end{figure}

We use a linear forward operator that observes $8$ evenly spaced vertical columns in a manner similar to \cite{farchi_bocquet_localization}. For each observed column, $5$ weighted vertical sums are measured, amounting to a data dimension of $d = 40$. For an observed column index $i$, these sums are given by
\begin{equation*}
I_{i,\, r} = \sum_{j = 1}^{N_z} X_{i,\, j} f_r(j),\quad f_r(j) = W_r^{-1} \exp\left( -\frac{(j - 6r)^2}{2b^2} \right),\,\quad 1 \leq r \leq 5,
\end{equation*}
with $b = 8$ and with $W_r$ chosen such that $\sum_j f_r(j)^2 = 1$. The weights $f_r(j)$ are plotted in \cref{fig:system}. Observations are perturbed with Gaussian noise vectors whose components are independent and zero-mean with variance $0.25$. This corresponds to approximately $1\%$ of the climatological variance for each channel. These observations are intended to mimic satellite radiance measurements, especially in the sense that they are vertically nonlocal.

\subsection{Localization}

\indent

To attenuate spurious correlations between state-vector indices $i$ and $j$, corresponding to 2D coordinates $(x_i,\, y_i)$ and $(x_j,\, y_j)$ with $1 \leq x_i,\, x_j \leq 40$ and $1 \leq y_i,\, y_j \leq 32$, we use the localizing value
\begin{equation*}
\mL(i,\, j) = \cG_2\left( \frac{1}{L} \sqrt{ c(x_i,\, x_j)^2 + (y_i - y_j)^2 } \right),
\end{equation*}
where $\cG_2$ is the univariate Gaspari-Cohn function \cite{gaspari_cohn} supported on $[-2,\, 2]$, $c(x,\, x')$ is chordal distance on a circle of circumference $N_h$ (cf.\ \cref{subsec:single_cycle_model}), and $L = 3$ is the localization length scale. This value for $L$ is selected by running the cycled DA experiment described below for several different values of $\ell$, and selecting the one which yielded the lowest average forecast error over the last $4000$ cycles. As in the first experiment, forecast errors for tuning $L$ are measured with respect to an idealized GETKF that forms the exact standard and modified Kalman gain matrices through ``brute force'' linear algebra.

\subsection{Experiment} \label{subsec:cycled_da_experiment}

\indent

To create an initial ground-truth state and forecast ensemble, we generate $41$ independent Gaussian random fields, and we integrate each field through \cref{eqn:ll96_ode} on the interval $0 \leq t \leq 10$. The first $m = 40$ results constitute the initial ensemble, and the final result is the initial ground-truth state. The same initial ensemble is used for all five filters. Each filter assimilates data from the ground-truth state into its ensemble every $\Delta t = 0.05$ units of model time for a total of $5000$ cycles. At each cycle, Krylov methods within the Krylov-based GETKF and InFo-ESRF are run for $10$ iterations, and all filters are run with posterior ensemble inflation in the form of relaxation to prior spread \cite{whitaker_rtps} with relaxation parameter $0.01$.

As in the first experiment, the augmentation methods and InFo-ESRF are evaluated over variations in a parameter $k$. For the augmentation methods, $k$ is the ratio of the augmented ensemble size over the original ensemble size, and for InFo-ESRF, $k$ is the number of quadrature nodes. We use the preconditioner described in \cref{section:preconditioning} for the mean-update stage of the Krylov-based GETKF, and for both the mean and perturbation updates in InFo-ESRF. The number of Ritz vectors for constructing the preconditioner is fixed at $p = 10$.

After $1000$ burn-in cycles, forecast error and forecast ensemble variance are recorded over the remaining cycles. We report two overall measures of filter performance:
\begin{align*}
\text{Average Forecast MSE} &= \frac{1}{4000} \sum_{T = 1001}^{5000} \frac{\| \tildemu_T - \vx_T \|_2^2}{ n }, \label{eqn:cycled_da_mse} \\
\text{Average Forecast }\frac{\text{MSE}}{\text{Variance}} &= \frac{1}{4000} \sum_{t = 1001}^{5000} \frac{\| \tildemu_T - \vx_T \|_2^2}{\sigma_T^2},
\end{align*}
where $n = 1280$ is the model dimension, $\tildemu_T$ is the forecast mean at cycle $T$, $\sigma_T^2$ is the forecast variance at cycle $T$ averaged over all components of the state vector, and $\vx_T$ is the ground-truth state at cycle $T$. For a successful filter, we expect the MSE-over-variance to be near unity.
\begin{figure}
    \centering
    \includegraphics[scale=.8]{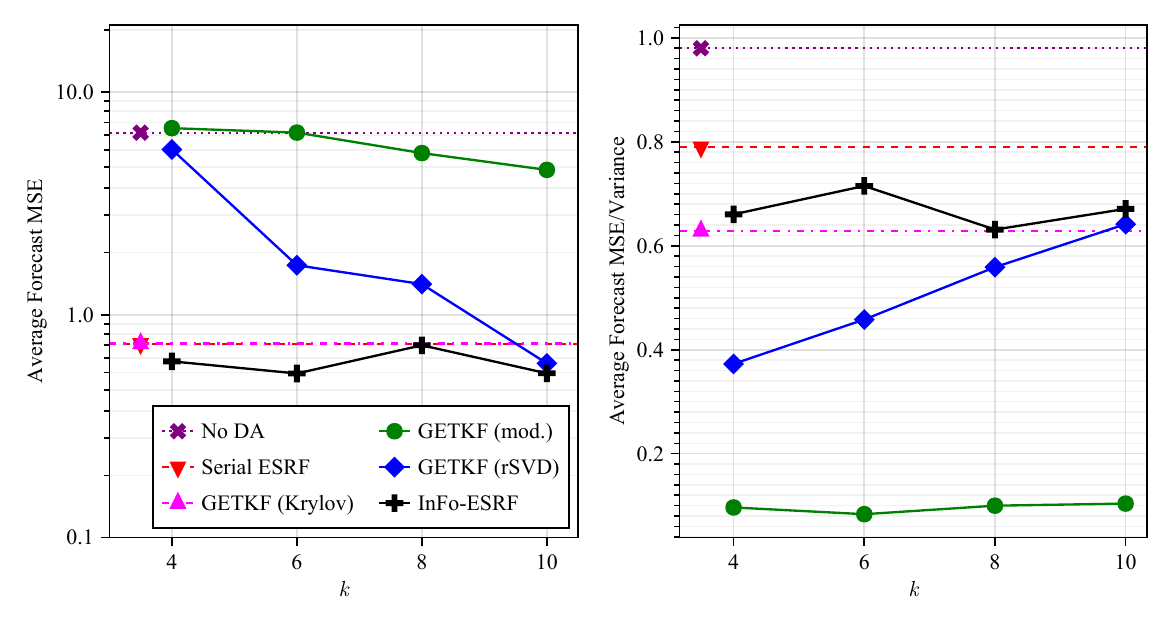}
    \caption{Results of the cycled data assimilation experiment. For the augmentation-based filters, $k$ is the ratio of the augmented ensemble size to the original ensemble size. For InFo-ESRF, $k$ is the number of quadrature nodes.}
    \label{fig:cycled_da}
\end{figure}

\subsection{Results}

\indent

The results of this experiment, repeated over 30 independent trials, are shown in \cref{fig:cycled_da}. For comparison, \cref{fig:cycled_da} also shows forecast errors and variance for an ensemble that does not undergo any data assimilation. For a given ensemble size, GETKF with a randomized SVD generally produces more accurate forecasts than GETKF with ensemble modulation, and has a better ratio of forecast MSE to variance. This is to be expected since, as explained in \ref{subsec:localized_filters}, the randomized SVD method produces a more optimal approximation of the localized covariance for a given number of augmented ensemble members. The Krylov-based GETKF performs better than either augmentation technique in this experiment, owing to the fact that its Krylov method directly approximates the action of the modified Kalman gain (rather than approximating the covariance ``upstream'' of it). InFo-ESRF performs nearly identically to the Krylov-based GETKF, and significantly, this occurs more-or-less independently of the quadrature size. This indicates that the use of quadrature to approximate the Kalman perturbation update not introduce significant error into InFo-ESRF, even for a very small number of quadrature nodes.
    \section{Conclusions}

\indent

We have demonstrated an integral-form ensemble square-root filter (InFo-ESRF) that is able to efficiently update ensemble perturbations in the presence of model space covariance localization. Our filter uses a discretized integral representation of the Kalman filter equations to update perturbations without directly evaluating a matrix square-root. This update procedure is easily parallelizable and uses Krylov methods in a manner that allows for preconditioning. Through numerical experiments with a synthetic forecast distribution and a multi-layer Lorenz system, we have demonstrated that InFo-ESRF has advantages in both accuracy and computational cost when compared to existing localized filters.

\section{Acknowledgments}

\indent

This paper was greatly improved by suggestions and insights from many of our colleagues. Jeff Whitaker (NOAA Physical Sciences Laboratory) helped us properly contextualize this work by making us aware of Steward et al.'s paper on Krylov-based ensemble Kalman filtering \cite{steward_matrix_functions}. Stephen Becker (University of Colorado Dept.\ of Applied Mathematics) and Anil Damle (Cornell University Dept.\ of Computer Science) both provided important references on quadrature methods for computing matrix square-roots. Alex Townsend (Cornell University Dept.\ of Mathematics) gave feedback on an initial presentation of this work which helped us improve our numerics. We are particularly indebted to Chris Snyder (NSF-NCAR Mesoscale and Microscale Meteorology), who hosted Robin Armstrong as a visitor at NSF-NCAR to begin developing an implementation of InFo-ESRF for full-scale weather forecasting problems.

Robin Armstrong was partially supported by the US Department of Energy Office of Science award DE-SC0025453. Ian Grooms was supported by the US National Science Foundation award 2152814.

    \bibliographystyle{siam}
    \bibliography{text/sources}
    
    \appendix

\section{Ensemble Adjustment Matrix} \label{section:adjustment_proof}

\indent

In the notation of \cref{section:background}, let $\mA = (\mSigma_a\mSigma_{xx}^{-1})^{1/2}$. Here we show that $\mA$ is a valid adjustment matrix, in the sense that it is well-defined  and $\mA\mSigma_{xx}\mA\tp = \mSigma_a$. For $\mA$ to be well-defined, it is necessary and sufficient that the eigenvalues of $\mSigma_a\mSigma_{xx}^{-1}$ are real and nonnegative. This is indeed the case, because
\begin{equation*}
\mSigma_{xx}^{-1/2}(\mSigma_a\mSigma_{xx}^{-1})\mSigma_{xx}^{1/2} = \mSigma_{xx}^{-1/2}\mSigma_a\mSigma_{xx}^{-1/2},
\end{equation*}
and the right-hand side is symmetric positive semidefinite. Taking the square-root on both sides, we find that
\begin{equation*}
\mSigma_{xx}^{-1/2} \mA\mSigma_{xx}^{1/2} = (\mSigma_{xx}^{-1/2}\mSigma_a\mSigma_{xx}^{-1/2})^{1/2},
\end{equation*}
and therefore,
\begin{align*}
\mA\mSigma_{xx}\mA\tp &= \mSigma_{xx}^{1/2} (\mSigma_{xx}^{-1/2}\mSigma_a\mSigma_{xx}^{-1/2})^{1/2} \mSigma_{xx}^{-1/2}\mSigma_{xx}\mSigma_{xx}^{-1/2} (\mSigma_{xx}^{-1/2}\mSigma_a\mSigma_{xx}^{-1/2})^{1/2} \mSigma_{xx}^{1/2} \\
&= \mSigma_{xx}^{1/2} \mSigma_{xx}^{-1/2}\mSigma_a\mSigma_{xx}^{-1/2} \mSigma_{xx}^{1/2} \\
&= \mSigma_a,
\end{align*}
meaning that $\mA$ provides the correct transformation.

\section{Derivation of the Integral Form Perturbation Update} \label{section:integral_proof}

\indent

In this section we prove \cref{theorem:stochastic_inflation,theorem:quadrature_trick}. Let $\vmu_a,\, \mSigma_a$ be the Kalman filter analysis given by \cref{eqn:kalman_filter}, and define
\begin{equation*}
\rvU_f = \begin{bmatrix}
\rvX_f \\
h(\rvX_f)
\end{bmatrix},\quad \mPi_x = \begin{bmatrix} \mI_{n \times n} & \mZero_{n \times d} \end{bmatrix}, \quad\text{and}\quad \mPi_h = \begin{bmatrix} \mZero_{d \times n} & \mI_{d \times d} \end{bmatrix},
\end{equation*}
so that $\rvX_f = \mPi_x\rvU_f$ and $\rvY = \mPi_h\rvU_f + \rvE_0$, recalling that $\rvE_0 \sim \cN(\vzero,\, \mR)$ is observation noise independent of $\rvX_f$. Let
\begin{equation*}
\mGamma_f \defeq \cov[\rvU_f] = \begin{bmatrix}
\mSigma_{xx} & \mSigma_{xh} \\
\mSigma_{xh}\tp & \mSigma_{hh} \\
\end{bmatrix}.
\end{equation*}
We will derive an integral formula for the Kalman filter analysis on $\rvU_f$ and $\rvY$, taking advantage of linearity in the observation model that relates the two, and projecting onto the state-space component will complete the proof. However, we face the difficulty that $\mGamma_f$ may be singular; this is the case when, for example, $h$ is affine. Thus, given $\varepsilon > 0$ we define $\rvU_{f,\, \varepsilon} = (\rvX_f,\, h(\rvX_f) + \varepsilon \rvE_1)$, where $\rvE_1 \sim \cN(\vzero,\, \mI)$ is independent of $\rvX_f$ and $\rvE_0$. Let $\mSigma_{hh,\, \varepsilon} = \mSigma_{hh} + \varepsilon\mI$, and note that
\begin{equation*}
\mGamma_{f,\, \varepsilon} \defeq \cov[\rvU_{f,\, \varepsilon}] = \begin{bmatrix}
\mSigma_{xx} & \mSigma_{xh} \\
\mSigma_{xh}\tp & \mSigma_{hh,\, \varepsilon} \\
\end{bmatrix}\,.
\end{equation*}
Because $\mGamma_f$ is positive semidefinite, a Schur complement identity \cite[eq 0.8.5.3]{horn_johnson} implies that $\mSigma_{hh} - \mSigma_{xh}\tp\mSigma_{xx}^{-1}\mSigma_{xh} \succeq \mZero$ in the positive semidefinite ordering. We therefore have $\mSigma_{hh,\, \varepsilon} - \mSigma_{xh}\tp\mSigma_{xx}^{-1} \mSigma_{xh} \succeq \varepsilon \mI \succ \mZero$, and by reapplying the same Schur complement identity, we find that $\mGamma_{f,\, \varepsilon}$ is strictly positive definite.

Let $\mGamma_a$, resp.\ $\mGamma_{a,\, \varepsilon}$, be the Kalman filter analysis covariance matrices for $\rvU_f$ and $\rvY$, resp.\ $\rvU_{f,\, \varepsilon}$ and $\rvY$. Referencing \cref{eqn:kalman_filter} and using $\cov[\rvU_f,\, h(\rvX_f)] = \mGamma_f\mPi_h\tp$ and $\cov[\rvU_{f,\, \varepsilon}$, $h(\rvX_f)] = \mGamma_{f,\, \varepsilon}\mPi_h\tp$, we have
\begin{align*}
\mGamma_a &= \mGamma_f - \mGamma_f\mPi\tp(\mR + \mSigma_{hh})^{-1}\mPi_h\mGamma_f = \lim_{\varepsilon \to 0} \mGamma_{f,\, \varepsilon} - \mGamma_{f,\, \varepsilon}\mPi_h\tp(\mR + \mSigma_{hh,\, \varepsilon})^{-1}\mPi_h\mGamma_{f,\, \varepsilon} \\
&= \lim_{\varepsilon \to 0} \mGamma_{a,\, \varepsilon}.
\end{align*}
Now we will relate $\mGamma_{f,\, \varepsilon}$ to $\mGamma_{a,\, \varepsilon}$ using the adjustment matrix derived in \cref{section:adjustment_proof}. This gives $\mGamma_{a,\, \varepsilon} = \mA_\varepsilon\mGamma_{f,\, \varepsilon}\mA_\varepsilon\tp$, where
\begin{align}
\mA_\varepsilon &= (\mGamma_{a,\, \varepsilon}\mGamma_{f,\, \varepsilon}^{-1})^{1/2} = (\mI - \mGamma_{f,\, \varepsilon}\mPi_h\tp(\mR + \mPi_h\mGamma_{f,\, \varepsilon}\mPi_h)^{-1}\mPi_h)^{1/2} \nonumber \\
&= (\mI + \mGamma_{f,\, \varepsilon}\mPi_h\tp\mR^{-1}\mPi_h)^{-1/2}, \label{eqn:epsilon_adjustment}
\end{align}
having used $\mSigma_{hh,\, \varepsilon} = \mPi_h\mGamma_{f,\, \varepsilon}\mPi_h\tp$ and the Sherman Morrison Woodbury identity.

Recall the definition $\mC = \mR^{-1/2}\mSigma_{hh}\mR^{-1/2}$, and let
\begin{equation*}
\mC_\varepsilon = \mR^{-1/2}\mSigma_{hh,\, \varepsilon}\mR^{-1/2} = \mC + \varepsilon\mR^{-1}.
\end{equation*}
In addition, recall that $r_t,\, s_t \geq 0$ are continuous functions of $t \in \cI$ such that
\begin{equation}
\int_{t \in \cI} \frac{r_t}{s_t + c + 1}\, dt = \frac{1}{\sqrt{c + 1}} \label{eqn:generic_invsqrt_integral}
\end{equation}
for all $c$ in some open set containing $\lambda(\mC) \union \{ 0 \}$. Because $\lambda(\mC_\varepsilon) \subseteq \lambda(\mC) \pm \varepsilon \| \mR^{-1} \|_2$, \cref{eqn:generic_invsqrt_integral} also holds for all $c \in \lambda(\mC_\varepsilon)$ when $\varepsilon$ is small enough. Finally, let
\begin{equation*}
\mF_\varepsilon = \mI + \mGamma_{f,\, \varepsilon}\mPi_h\tp\mR^{-1}\mPi_h,\quad \mG_\varepsilon = \mR^{-1/2}\mPi_h\mGamma_{f,\, \varepsilon}^{1/2}.
\end{equation*}
\Cref{eqn:epsilon_adjustment} gives $\mA_\varepsilon = \mF_\varepsilon^{-1/2}$, and by the diagonalizability of $\mG_\varepsilon\tp\mG_\varepsilon$,
\begin{equation}
\mF_\varepsilon = \mGamma_{f,\, \varepsilon}^{-1/2}(\mI + \mG_\varepsilon\tp\mG_\varepsilon)\mGamma_{f,\, \varepsilon}^{1/2} \label{eqn:adjustment_symmetrization}
\end{equation}
is also diagonalizable. Let $\lambda_i(\cdot)$ denote the $i\nth$ largest eigenvalue of a positive-definite matrix, with $\lambda_i(\cdot) \defeq 0$ when $i$ exceeds the matrix dimension. For all $i \geq 1$, \cref{eqn:adjustment_symmetrization} implies that
\begin{align*}
\lambda_i(\mF_\varepsilon) = 1 + \lambda_i(\mG_\varepsilon\tp\mG_\varepsilon) = 1 + \lambda_i(\mG_\varepsilon\mG_\varepsilon\tp) = 1 + \lambda_i(\mC_\varepsilon).
\end{align*}
We now see that the matrices $r_t(s_t\mI + \mF_\varepsilon)^{-1},\, t \in \cI$, are simultaneously diagonalizable with eigenvalues $r_t (s_t + \lambda_i(\mC_\varepsilon) + 1)^{-1}$. Because of this, we may apply \cref{eqn:generic_invsqrt_integral} eigenvalue-wise to $\mF_\varepsilon$, yielding
\begin{equation*}
\mA_\varepsilon = \mF_\varepsilon^{-1/2} = \int_{t \in \cI} r_t(s_t\mI + \mF_\varepsilon)^{-1}\, dt = \int_{t \in \cI} r_t((s_t + 1)\mI + \mGamma_{f,\, \varepsilon}\mPi_h\tp\mR^{-1}\mPi_h)^{-1}\, dt.
\end{equation*}
Using the Sherman-Morrison-Woodbury identity,
\begin{align*}
\mA_\varepsilon &= \int_{t \in \cI} r_t \left( \frac{1}{s_t + 1}\mI - \frac{1}{(s_t + 1)^2}\mGamma_{f,\, \varepsilon}\mPi_h\tp(\mR + (s_t + 1)^{-1}\mPi_h\mGamma_{f,\, \varepsilon}\mPi_h\tp)^{-1}\mH \right)\, dt \\
&= \int_{t \in \cI} p_t\mI\, dt -  \int_{t \in \cI} p_t \mGamma_{f,\, \varepsilon}\mPi_h\tp((s_t + 1)\mR + \mSigma_{hh,\, \varepsilon})^{-1}\mPi_h \, dt,
\end{align*}
where $p_t = r_t(s_t + 1)^{-1}$. Substituting $c = 0$ into \cref{eqn:generic_invsqrt_integral} shows that $\int_{t \in \cI} p_t\, dt = 1$, and correspondingly $\int_{t \in \cI} p_t \mI\, dt = \mI$.

We will now obtain an integral formula for $\mSigma_a$. Using $\mGamma_f = \lim_{\varepsilon \to 0} \mGamma_{f,\, \varepsilon}$ and $\mGamma_a = \lim_{\varepsilon \to 0} \mGamma_{a,\, \varepsilon}$,
\begin{equation*}
\mSigma_a = \mPi_x\mGamma_a\mPi_x\tp = \mPi_x\left( \lim_{\varepsilon \to 0} \mA_\varepsilon\mGamma_{f,\, \varepsilon}\mA_\varepsilon\tp \right)\mPi_x\tp = \mA_x\mGamma_f\mA_x\tp,
\end{equation*}
where
\begin{align}\notag
\mA_x &= \lim_{\varepsilon \to 0} \mPi_x\mA_\varepsilon = \mPi_x -  \int_{t \in \cI} p_t \mPi_x\mGamma_f\mPi_h\tp((s_t + 1)\mR + \mSigma_{hh})^{-1}\mPi_h \, dt \\
&= \mPi_x - \left( \int_{t \in \cI} p_t\mK(s_t)\, dt \right) \mPi_h \label{eqn:adjustment_integral_form},
\end{align}
having used $\mPi_x\mGamma_f\mPi_h\tp = \mSigma_{xh}$ and \cref{eqn:inflated_kalman_gain}. Because $\cov[\rvU_h] = \mGamma_f$, defining $\rvZ_a = \mA_x(\rvU_f - \vmu)$ implies $\E[\rvZ_a] = \vzero$ and $\cov[\rvZ_a] = \mSigma_a$. Inserting \cref{eqn:adjustment_integral_form} for $\mA_x$ and using $\mPi_x (\rvU_f - \vmu) = \rvZ_f,\, \mPi_h(\rvU_f - \vmu) = \rvW_f$, we have
\begin{equation*}
\rvZ_a = \rvZ_f - \left(\int_{t \in \cI} p_t\mK(s_t)\, dt\right)\rvW_f,
\end{equation*}
which proves \cref{theorem:quadrature_trick}.

\Cref{theorem:stochastic_inflation} is a corollary of \cref{theorem:quadrature_trick} that arises from a specific choice of $\cI$, $r_t$, and $s_t$. Starting with $\pi / 2 = \int_0^\infty (1 + x^2)^{-1}\, dx$ and making the substitution $x = (c+1)^{-1/2} \tan\left( \frac{\pi s}{2} \right)$, where $c \in (-1,\, \infty)$ is arbitrary, we have
\begin{equation}
\frac{1}{\sqrt{c + 1}} = \int_0^1 \frac{\sec^2\left( \frac{\pi s}{2} \right)}{\tan^2\left( \frac{\pi s}{2}\right) + c + 1}\, ds. \label{eqn:trig_quadrature_rule}
\end{equation}
Now we let $s = 2\pi^{-1}\arctan(\sqrt{t})$, so that $\tan^2\left( \frac{\pi s}{2}\right) = t$ and $\sec^2\left( \frac{\pi s}{2}\right) = t + 1$. This gives
\begin{equation*}
\frac{1}{\sqrt{c + 1}} = \int_0^\infty \frac{1}{\pi \sqrt{t}(t + 1 + c)}\, dt.
\end{equation*}
The hypotheses of \cref{theorem:quadrature_trick} are therefore satisfied by
\begin{equation}
\cI = (0,\infty),\quad r_t = \frac{1}{\pi\sqrt{t}},\quad s_t = t,\quad p_t = \frac{r_t}{s_t + 1} = \frac{1}{\pi \sqrt{t}(t+1)}. \label{eqn:generic_quadrature}
\end{equation}
Inserting these expressions into \cref{theorem:quadrature_trick} yields \cref{theorem:stochastic_inflation}.

\section{Integral Form of the Modified Kalman Gain Matrix} \label{section:modified_kalman_gain}

\indent

Here we prove \cref{prop:integral_getkf_equivalence}, which states that if $\cI,\, r_t,\, s_t$, and $p_t$ satisfy the assumptions of \cref{theorem:quadrature_trick}, then
\begin{equation*}
\int_{t \in \cI} p_t \mK(s_t)\,dt = \mG,
\end{equation*}
where $\mG = \mSigma_{xh}(\mR + \mSigma_{hh} + \mR(\mI + \mR^{-1}\mSigma_{hh})^{1/2})^{-1}$ is the modified Kalman gain matrix from \cref{eqn:modified_gain}. We start by rewriting this integral using a spectral decomposition $\mR^{-1/2}\mSigma_{hh}\mR^{-1/2} = \mQ\mLambda\mQ\tp$. This gives
\begin{align*}
\int_{t \in \cI} p_t \mK(s_t)\, dt &= \mSigma_{xh}\mR^{-1/2}\mQ\left( \int_{t \in \cI} p_t ((s_t + 1)\mI + \mLambda)^{-1}\, dt \right)\mQ\tp\mR^{-1/2}.
\end{align*}
The matrix $\int_{t \in \cI} p_t ((s_t + 1)\mI + \mLambda)^{-1}\, dt$ is diagonal, and we use a partial fraction decomposition to express its $(i,i)$ entry:
\begin{align}
\int_{t \in \cI} \frac{p_t}{s_t + 1 + \lambda_i}\, dt &= \int_{t \in \cI} \frac{r_t}{(s_t + 1)(s_t + 1 + \lambda_i)}\,dt = \frac{1}{\lambda_i} \int_{t \in \cI} \left( \frac{r_t}{s_t + 1} - \frac{r_t}{s_t + 1 + \lambda_i} \right)\, dt.
\end{align}
The definition of $r_t$ and $s_t$ implies that
\begin{equation*}
\frac{1}{\lambda_i} \int_{t \in \cI} \left( \frac{r_t}{s_t + 1} - \frac{r_t}{s_t + 1 + \lambda_i} \right)\, dt = \frac{1}{\lambda_i}\left( 1 - \frac{1}{\sqrt{1 + \lambda_i}} \right) = \frac{1}{1 + \lambda_i + \sqrt{1 + \lambda_i}}.
\end{equation*}
We now have $\int_{t \in \cI} p_t ((s_t + 1)\mI + \mLambda)^{-1}\, dt = (\mI + \mLambda + (\mI + \mLambda)^{1/2})^{-1}$, and
\begin{align*}
\int_{t \in \cI} p_t \mK(s_t)\, dt &= \mSigma_{xh}\mR^{-1/2}\mQ(\mI + \mLambda + (\mI + \mLambda)^{1/2})^{-1} \mQ\tp\mR^{-1/2} \\
&= \mSigma_{xh}(\mR + \mSigma_{hh} + \mR^{1/2}(\mI + \mR^{-1/2}\mSigma_{hh}\mR^{-1/2})^{1/2}\mR^{1/2})^{-1}.
\end{align*}
Note that
\begin{align*}
\mR^{1/2}(\mI + \mR^{-1/2}\mSigma_{hh} \mR^{-1/2})^{1/2}\mR^{1/2} &= \mR\mR^{-1/2}(\mI + \mR^{-1/2}\mSigma_{hh}\mR^{-1/2})^{1/2}\mR^{1/2} \\
&= \mR\left[ \mR^{-1/2}(\mI + \mR^{-1/2}\mSigma_{hh}\mR^{-1/2})\mR^{1/2} \right]^{1/2} \\
&= \mR(\mI + \mR^{-1}\mSigma_{hh})^{1/2}.
\end{align*}
Thus $\int_{t \in \cI} p_t \mK(s_t)\, dt = \mSigma_{xh}(\mR + \mSigma_{hh} + \mR(\mI + \mR^{-1}\mSigma_{hh})^{1/2})^{-1} = \mG$, completing the proof.

\section{Quadrature Rules} \label{section:quad_rule}

\indent

To discretize the integral in InFo-ESRF, we use a quadrature rule for matrix square-roots devised by Hale, Higham, and Trefethen \cite{hale_higham_contour_trick}  using contour integration methods. Their work, rephrased in the notation of this paper, shows that \cref{eqn:generic_invsqrt_integral} holds with
\begin{align*}
\cI = (0,\, \cK(g)),\quad r_t = \pi^{-1} |\mathrm{cn}(it | g) \mathrm{dn}(it | g)|,\quad s_t = |\mathrm{sn}(it | g)|^2,
\end{align*}
where $i$ is the imaginary unit, $\cK$ is the complete elliptic integral of the first kind, $g = \sqrt{1 - (1 + \ell)^{-1}}$, $\ell$ is any value in $(\lambda_\mathrm{max}(\mC),\, \infty)$, and $\mathrm{sn},\, \mathrm{cn},\, \mathrm{dn}$ are Jacobi elliptic functions. Our implementation of InFo-ESRF uses a rough order-of-magnitude estimate for $\ell$ that is chosen beforehand and reused for each assimilation cycle. A trapezoid rule discretization yields
\begin{equation}
\begin{split}
s^{(q)} &= |\mathrm{sn}(it_q \cK(g) \,|\, g)|^2,\quad r^{(q)} = \frac{2}{\pi Q} |\cK(g)\mathrm{cn}(it_q \cK(g) \,|\, g)\mathrm{dn}(iu_q \cK(g) \,|\, g)|, \\
w^{(q)} &= r^{(q)}(s^{(q)} + 1)^{-1},
\end{split} \label{eqn:elliptic_quadrature_rule}
\end{equation}
for $q = 1,\, \ldots,\, Q$, where $Q$ is the quadrature size and $t_q = Q^{-1}(q - \frac{1}{2})$. We use this quadrature rule because of its extremely fast convergence rate. Indeed, to reach a desired error tolerance $\varepsilon > 0$, it requires only $Q = \cO((\log \varepsilon^{-1})(\log \ell))$ nodes \cite[Thm 4.1]{hale_higham_contour_trick}.

While our data assimilation experiments use the above quadrature rule, which we call the ``elliptical quadrature,'' other methods are possible. As an example, in light of \cref{eqn:trig_quadrature_rule}, we consider discretizing an integral form based on $\cI = (0,1),\, r_t = \sec^2\left( \frac{\pi t}{2} \right),\, s_t = \tan^2\left( \frac{\pi t}{2} \right),\, p_t \equiv 1$ using Gauss-Legendre quadrature. This yields
\begin{equation}
s^{(q)} = \tan^2\left( \frac{\pi \lambda_q}{2} \right),\quad p^{(q)} = \frac{v_q}{2} \label{eqn:gaussian_quadrature_rule}
\end{equation}
for $q = 1,\, \ldots,\, Q$, where $v_q,\, \lambda_q$ are Gauss-Legendre weights and nodes that can be computed with the Golub-Welsch algorithm \cite{golub_welsch}.

\Cref{fig:quad_rule_convergence} shows the convergence of these two quadrature rules for approximating a scalar modified Kalman gain using the integral representation from \cref{prop:integral_getkf_equivalence}:
\begin{equation*}
G = \frac{\sigma_{xh}}{1 + \sigma_{hh} + \sqrt{1 + \sigma_{hh}}} \approx \sum_{q = 1}^Q \frac{p^{(q)}\sigma_{xh}}{1 + s^{(q)} + \sigma_{hh}},
\end{equation*}
where $\sigma_{xh} = 20$ and $\sigma_{hh} = 10$. While the Gaussian quadrature rule is easier to implement using the standard functions of most scientific computing libraries, the elliptical quadrature rule converges much faster.

\begin{figure}
    \centering
    \includegraphics[scale=.85]{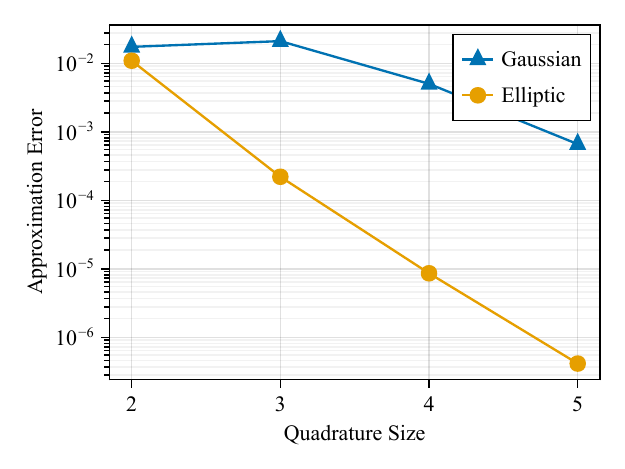}
    \caption{Convergence rates of the Gaussian quadrature rule, \eqref{eqn:gaussian_quadrature_rule}, and the elliptical quadrature rule, \eqref{eqn:elliptic_quadrature_rule}, for approximating a scalar modified Kalman gain.} \label{fig:quad_rule_convergence}
\end{figure}

\section{Linear System Solves} \label{section:preconditioning}

\indent

As explained in \cref{section:algorithm}, InFo-ESRF computes analysis perturbations by solving the linear systems
\begin{equation*}
((s^{(q)} + 1)\mR + \widehat{\mSigma}_{hh})\vv_q^{\,(i)} = \vw_f^{\,(i)}
\end{equation*}
for $1 \leq i \leq m$ and $1 \leq q \leq Q$. We compute $\vv_q^{\,(i)}$ by first solving
\begin{equation}
\widehat{\mC}_q \vu_q^{\,(i)} = \mR^{-1/2}\vw_f^{\,(i)} \label{eqn:normalized_linear_solve}
\end{equation}
where $\widehat{\mC}_q = (s^{(q)} + 1)\mI + \mR^{-1/2}\widehat{\mSigma}_{hh}\mR^{-1/2}$, and then setting $\vv_q^{\,(i)} = \mR^{-1/2}\vu_q^{\,(i)}$. \Cref{eqn:normalized_linear_solve} is solved with a preconditioned conjugate gradients iteration \cite[Alg 11.5.1]{golub_van_loan}. We use a limited-memory preconditioner \cite{tshimanga_preconditioning}, denoted $\mP_q$ and defined by
\begin{equation*}
\mP_q^{-1} = (\mI - \mPhi_q\mTheta_q^{-1}\mPhi_q\tp\widehat{\mC}_q)(\mI - \widehat{\mC}_q\mPhi_q\mTheta_q^{-1}\mPhi_q\tp) + \beta \mPhi_q\mTheta_q^{-1}\mPhi_q\tp,
\end{equation*}
where $\beta$ is a scalar (see below), $\mPhi_q = [\vphi_q^{\,(1)}\:\cdots\:\vphi_q^{\,(p)}]$, $\mTheta = \mathrm{diag}(\theta_q^{(1)},\, \ldots,\, \theta_q^{(p)})$, and $(\vphi_q^{\,(j)}, \theta_q^{(j)})$ are Ritz pairs \cite[sec 10.1.4]{golub_van_loan} corresponding to a $p$-dimensional approximate eigenspace of $\widehat{\mC}_q$ found by a randomized symmetric eigendecomposition \cite[Alg 5.3]{halko_finding_structure_with_randomness}. Because $\widehat{\mC}_1,\, \ldots,\, \widehat{\mC}_Q$ differ by scalar shifts of $\mI$, $(\vphi_q^{\,(j)},\, \theta_q^{(j)})$ for $q \geq 2$ can be found by shifting and rescaling $(\vphi_1^{\,(j)},\, \theta_1^{(j)})$.

This preconditioner endows $\mP_q^{-1/2}\widehat{\mC}_q\mP_q^{-1/2}$ with a cluster of $p$ eigenvalues at $\beta$, and all other eigenvalues remain in the interval $W(\widehat{\mC}_q) \defeq [\lambda_\mathrm{min}(\widehat{\mC}_q),\, \lambda_\mathrm{max}(\widehat{\mC}_q)]$ \cite{tshimanga_preconditioning}. To ensure that the eigenvalue cluster at $\beta$ does not degrade the condition number, we use $\beta = \min_i \widehat{\mC}_q(i,i)$, a choice which guarantees that $\beta \in W(\widehat{\mC}_q)$.

\section{Implementation Notes} \label{section:filter_implementation}

\indent

We have done our best to implement the most efficient code possible for each filter considered in this paper. Great care was taken to avoid unnecessary memory allocations that would affect runtime measurements, including by preallocating space for ensemble means and perturbations outside of filter codes. Certain preprocessing and postprocessing operations common to all filters were also performed outside of filter codes so as not to be counted in runtime comparisons. Wherever possible, we avoided operations that explicitly loop over data components and ensemble members, and instead used block operations over vectors of data and matrices of ensemble members. Parallelism was incorporated implicitly through the multithreading already present in BLAS libraries for matrix and vector operations.

Our implementations of the serial ESRF \cite{shlyaeva_serial_filter} and Krylov-based GETKF \cite{steward_matrix_functions} differ slightly from the algorithms originally described by those authors. The differences are as follows.
\begin{enumerate}
\item Our implementation of the Krylov-based GETKF uses preconditioned conjugate gradients (PCG) iteration to update the mean, and Lanczos iteration with reorthogonalization to update the perturbations. For the mean update, we use the preconditioner described in \cref{section:preconditioning}. The original description of this algorithm \cite{steward_matrix_functions} uses a different Krylov method, namely a restarted Arnoldi iteration, for these steps. We have used a different Krylov method for the sake of convenience, and we have no reason to believe that our choices would degrade this filter in terms of accuracy or cost.

\item For computing prior observable variances in the serial ESRF \cite{shlyaeva_serial_filter}, we use a formula which is different from (though mathematically equivalent to) the one given in Shlyaeva et al.'s original paper. This was done for the sake of efficiency, in order to a loop over state components.
\end{enumerate}
Our filter code is available at \url{https://github.com/robin-armstrong/info-esrf-experiments}.
\end{document}